\documentclass[]{article}
\usepackage{amsmath,fullpage,calc,enumerate}
\usepackage{amsthm}
\usepackage{amssymb}
\usepackage{algorithmic}
\usepackage[ruled,vlined]{algorithm2e}
\usepackage{rotating}	
\usepackage{xcolor}
\usepackage{url}
\usepackage{cite}
\usepackage{verbatim}
\usepackage[subrefformat=parens,labelformat=parens]{subcaption}
\usepackage{graphicx,psfrag,multicol}
\DeclareGraphicsExtensions{.jpg,.eps,.mps,.png}
\allowdisplaybreaks

\newtheorem{theorem}{Theorem}

\newtheorem{lemma}[theorem]{Lemma}
\newtheorem{corollary}[theorem]{Corollary}

\theoremstyle{definition}

\begin{document}

\title{The Quadratic Minimum Spanning Tree Problem and its Variations}

\author{\sc{Ante \'Custi\'c}\thanks{{\tt acustic@sfu.ca}. Department of Mathematics, Simon Fraser University Surrey,  250-13450 102nd AV, Surrey, British Columbia, V3T 0A3, Canada}
\and
\sc{Ruonan Zhang}\thanks{{\tt ruonan.zhang@xjtlu.edu.cn}. Department of Mathematical Sciences, Xi'an Jiaotong-Liverpool University, 111 Ren'ai Road, Suzhou, Jiangsu, 215123, China}
\and
\sc{Abraham P. Punnen}\thanks{{\tt apunnen@sfu.ca}. Department of Mathematics, Simon Fraser University Surrey,  250-13450 102nd AV, Surrey, British Columbia, V3T 0A3, Canada}}

\maketitle

\begin{abstract}
The quadratic minimum spanning tree problem and its variations such as the quadratic bottleneck spanning tree problem, the minimum spanning tree problem with conflict pair constraints, and the bottleneck spanning tree problem with conflict pair constraints are useful in modeling various real life applications. All these problems are known to be NP-hard. In this paper, we investigate these problems to obtain additional insights into the  structure of the problems and to identify possible demarcation between easy and hard special cases. New polynomially solvable cases have been identified, as well as NP-hard instances on very simple graphs. As a byproduct, we have a recursive formula for counting the number of spanning trees on a $(k,n)$-accordion and a characterization of matroids in the context of a quadratic objective function.
\medskip

\end{abstract}

\section{Introduction}\label{secIntro}

Let $G=(V,E)$ be an undirected graph with $|E|=m$. Costs $c_e$ and $q(f,g)$ are given for each
edge $e\in E$ and each pair of edges $f,g\in E,\ f\neq g$, respectively. Then the
\textsl{quadratic minimum spanning tree problem (QMST)} is
formulated as follows:

\begin{tabbing}
\hspace{3cm}
\=xxxx\=xxxxx\=xxxxx\=xxxxx\=xxxxx\=xxxx\=xxxx\kill
\> Minimize $z(T)=\sum\limits_{e\in T}\sum\limits_{\substack{f\in T\\f\neq e}}q(e,f)+\sum\limits_{e\in T}c_e$\\
\> Subject to\\
\>\>\> $T\in \mathcal{F}$,
\end{tabbing}
where $\mathcal{F}$ is the family of all spanning trees of $G$. The associated \textsl{cost matrix} $Q_{m\times m}$ have its ($i,j$)-th entry as $q(i,j)$ when $i\neq j$, and as $c_i$ when $i=j$.

The QMST can be viewed as a generalization of many well known optimization problems such as the travelling salesman
problem, the quadratic assignment problem, the maximum clique problem etc., and it can be used in modeling various real life application areas such as telecommunication, transportation, irrigation energy distribution, and so on. The problem was introduced by Assad and Xu \cite{Assad}, along with its special case - \textsl{the adjacent-only quadratic minimum spanning tree problem (AQMST)}, in which $q(e,f)=0$ if $e$ and $f$ are not adjacent. The strong NP-hardness of both the QMST and AQMST was proved in~\cite{Assad} along with ideas for solving the problem using exact and heuristic algorithms. The broad applications base and inherent complexity of the QMST makes it an interesting topic for further research. Most of the works on QMST have been focussed on heuristic algorithms ~\cite{x2,x8,x3,Oncan,x1,x6,Soak,Soak1,Zhou}. \'Custi\'c and Punnen~\cite{c1} provided a characterization QMST instances that can be solved as a minimum spanning tree problem. Exact algorithm for AQMST and QMST was studied by Pereira, Gendreau, and Cunha~\cite{Assad,x6,x9,xx9}. A special case of QMST with one quadratic term was studied by Buchheim  and Klein~\cite{x5}, A. Fischer and F. Fischer~\cite{x4}.

There are many other problems which have been studied in the literature, that are closely related to the QMST in terms of formulation and applications. We list below some of these variations of QMST that we investigate in this paper.

\textsl{The minimum spanning tree problem with conflict pairs (MSTC)}~\cite{Darmann, da, Zh}: Given a graph $G$ with edge costs
$c_e$ and a set $S\subseteq \{\{e,f\}\subseteq E\colon  e\neq f\}$, the MSTC is to find a
spanning tree of $G$ such that the tree
cost is minimized and for each edge pair $\{e,f\}$ in $S$, at most
one of them is included in the tree. The feasibility problem of
the MSTC, denoted by FSTC, is the problem of finding a feasible solution of the MSTC, regardless
of the costs. Given an FSTC instance, we construct a QMST on the same graph with
\[
q(e,f)=
\begin{cases}
	1 & \mbox{ if } \{e,f\}\in S \\
	0 & \mbox{ otherwise}.
\end{cases}
\]
Then the FSTC instance is feasible if and only if the above QMST has the optimal objective function value 0. Therefore, the FSTC reduces to the QMST.

\textsl{The quadratic bottleneck spanning tree problem (QBST)}: By replacing the objective function of the QMST with $\max\{q(e,f):e,f\in T\}$, we obtain the QBST. The problem is introduced in
\cite{Zhang} and shown to be NP-hard even on a bipartite graph with 0-1 $q(e, f)$ values. FQBST, the feasibility version of the QBST, is described as \textsl{``Given a value $\mu$, does there
exist a spanning tree $T$ of $G$ such that $\max\limits_{e,f\in
T}q(e,f)\leq \mu$?"}. As the FQBST is equivalent to the FSTC, \cite{Zhang} develops heuristic algorithms for the QBST using MSTC heuristics as subroutines.

\textsl{The bottleneck spanning tree problem with conflict pairs (BSTC)}: Similar to the relation between the QBST and the QMST, the BSTC is defined by substituting the ``min-sum" objective function in the MSTC with a ``min-max" objective function.

Furthermore, we define the ``adjacent-only counterparts" for the above problems: AQBST, in which $q(e,f)=0$ if $e$ and $f$ are not adjacent in the graph $G$; MSTAC, FSTAC, BSTAC, where the edges in the conflict pairs are all restricted to be adjacent.
\medskip

Even though the above problems are all proved to be NP-hard in general, exploring nicely solvable special cases  provide additional insights into the structure of these problems and opportunities for developing effective heuristics~\cite{m1,m2}. The primary research question we focus in this paper is: \textsl{To what extend the QMST and its variations would retain its NP-hardness status, or become amenable for polynomial time  solvability?} We consider restricting the structure of the graph $G$ and that of the cost matrix $Q$ to identify possible demarkation between easy and hard instances.

The rest of the paper is organized as follows: Section~\ref{secGraphs} introduces the sparse graphs that we are investigating. These include fans, fan-stars, ladders, wheels and their generalizations, $(k,n)$-ladders and $(k,n)$-accordions. A recursive formula is derived to count the number of spanning trees of a $(k,n)$-accordion, which generalizes the well known sparring tree counting formulas for fans and ladders. In Section~\ref{secComplexity} we study the complexity of  QMST and its variations on these sparse graphs. It is shown that the problems are NP-hard, except in the case of AQBST, MSTAC, BSTAC and AQBST on $(k,n)$-ladders and for these cases, we provide  $O(kn)$ time algorithms. The problems on a general graph $G$ but with specially structured cost matrices  are discussed in Section~\ref{secGraded}. In particular, we show that when $Q$ is a permuted doubly graded matrix, QMST and QBST are solvable in polynomial time.
In this case the optimal solution value attains the Assad-Xu lower bound \cite{Assad}. This result is extended to the case of  matroid bases and it provides a new characterization of matroids.

We use the notations $V(G)$ and $E(G)$ to denote, respectively, the node and edge sets of a graph $G$.
The quadratic costs $q(e_i, e_j)$ for the edge-pair $(e_i,e_j)$ is sometimes denoted by $q(i,j)$ for simplicity.

\section{The $(k,n)$-accordion and the number of spanning trees}\label{secGraphs}

In this section we define the classes of graphs called $(k,n)$-ladders and $(k,n)$-accordions and study  QMST and its variations on these graphs. We also study the number of spanning trees on such graphs, that is, the number of feasible solutions of the corresponding QMST.

\begin{figure}[h]
     	\centering
	\includegraphics[scale=0.55]{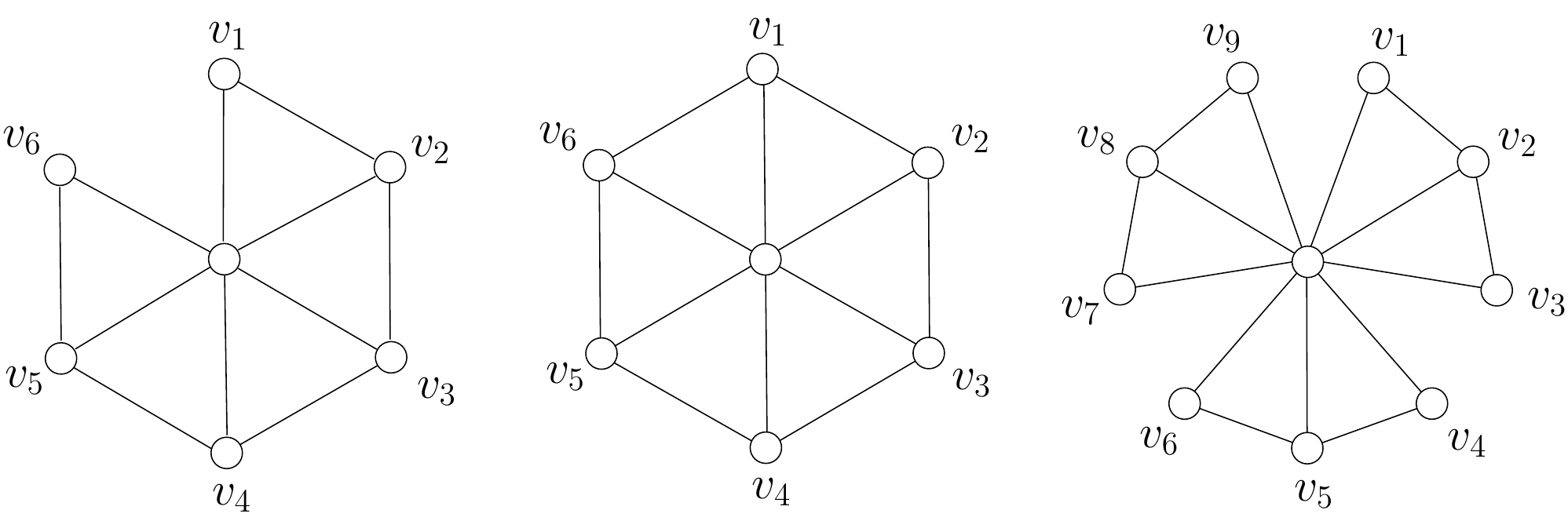}
    	\caption{A fan $F_6$, a wheel $W_6$ and a fan-star $FS_9$}
	\label{fan}
\end{figure}
Given a path $P_n=v_1-v_2-\cdots-v_n$, \textsl{a fan ($F_n$)} is
obtained by introducing a new node $u$ and edges $(u,v_i)$ for
$i=1,2,\ldots,n$. If we add one more edge $(v_1,v_n)$ to $F_n$,
the resulting graph is called \textsl{a wheel}, denoted by $W_n$. When
$n=3k$, deleting the edges $(v_{3i},v_{3i+1}),$
$i=1,\ldots,k-1,$ from $F_n$ results in \textsl{a fan-star
($FS_n$)}. Examples of a fan, a wheel and a fan-star are presented in
Figure~\ref{fan}.

Let $P_n^1=v_1-v_2-\cdots -v_n$ and
$P_n^2=u_1-u_2-\dots-u_n$ be two node-disjoint paths.  Add $n$ edges $(v_i,u_i),$
$i=1,\ldots, n,$ and the resulting graph is called a
\textsl{ladder ($L_n$)}, see Figure~\ref{ladder}.
\begin{figure}[h]
   	\centering
	\includegraphics[scale=0.6]{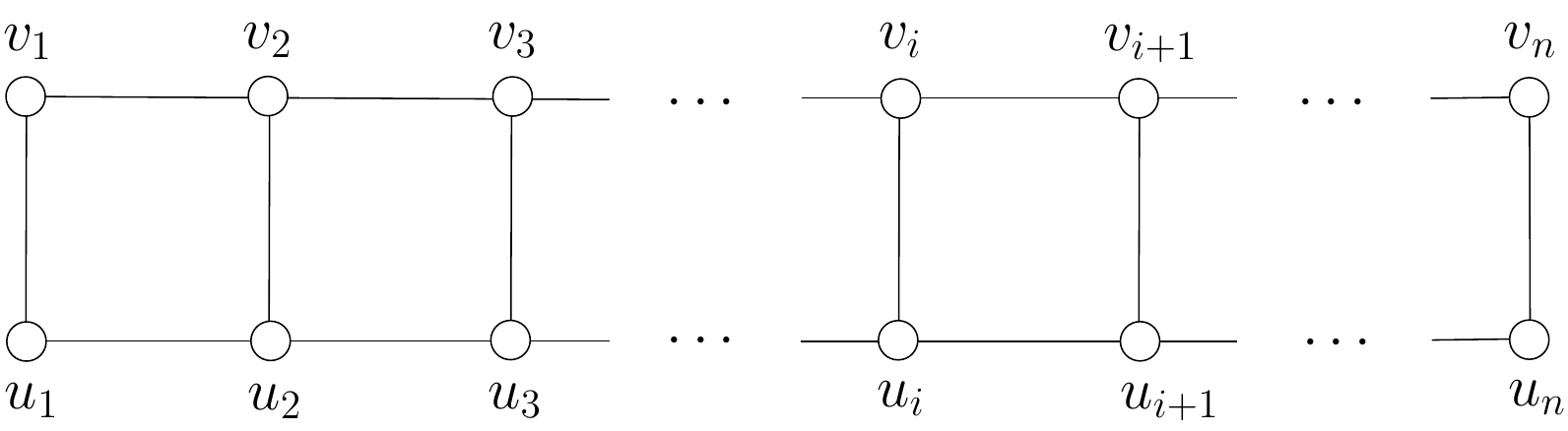}
    	\caption{A ladder $L_n$}
	\label{ladder}
\end{figure}

Let us now define a class of graphs that generalizes both fans and ladders. Given integers $k\geq 3$ and $n\geq 1$, a \textsl{$(k,n)$-accordion} is a  graph $A$ that can be obtained by recursively fusing together $n$ node-disjoint $k$-cycles $C^k_1,C^k_2,\ldots,C^k_n$ along an edge so that in the resulting graph, only two consecutive cycles have an edge in common. More precisely, a $(k,n)$-accordion  is a graph constructed using the following rules:
\begin{enumerate}
	\item[(i)] Initialize the graph $A_1$ as the $k$-cycle $C^k_1$. Embed $A_1$ on the plane and designate all its edges as free edges.
	\item[(ii)] For $i= 2$ to $n$ define $A_i$ as follows:
        Choose a free edge $(r,s)$ of $A_{i-1}$. Introduce a $k$-cycle, say $C^k_i$ to $A_{i-1}$ using the edge $(r,s)$ and $k-2$ new nodes so that we get a planar embedding of the resulting graph $A_i$. Designate any edge incident to a new node of $A_i$ as a free edge.
\end{enumerate}
Every so obtained graph $A_{n}$ is a $(k,n)$-accordion, and the set of all $(k,n)$-accordions is denoted by $A(k,n)$. 
Figure~\ref{knladder} presents examples from $A(5,7)$ and $A(4,9)$.\begin{figure}[h]
	\centering
    \begin{subfigure}[b]{0.4\textwidth}
	\centering
	\includegraphics[scale=0.53]{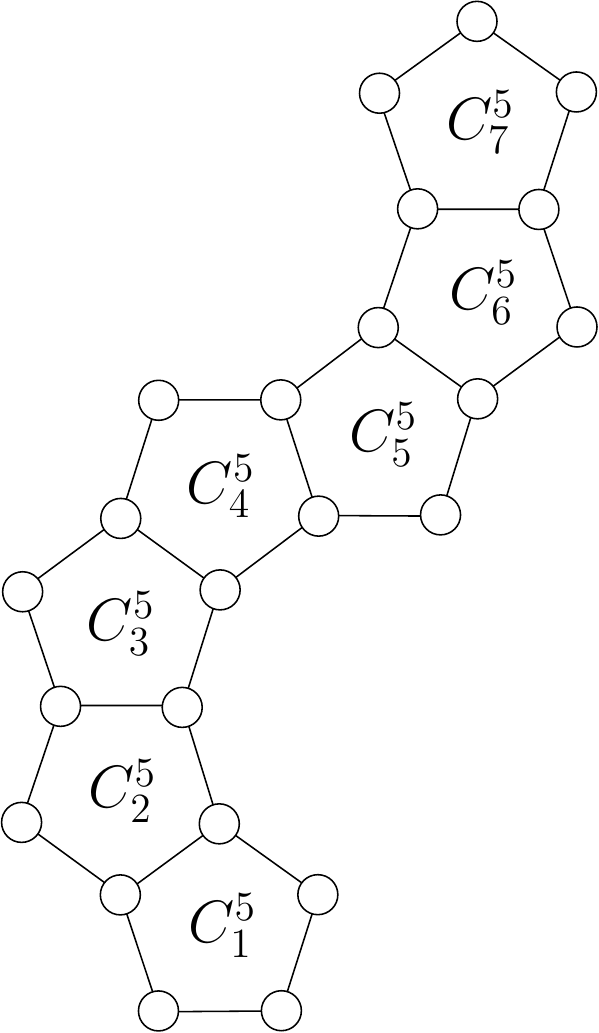}
        \caption{}
        \label{fig:ladd}
    \end{subfigure}
    \begin{subfigure}[b]{0.4\textwidth}
	\centering
		\includegraphics[scale=0.55]{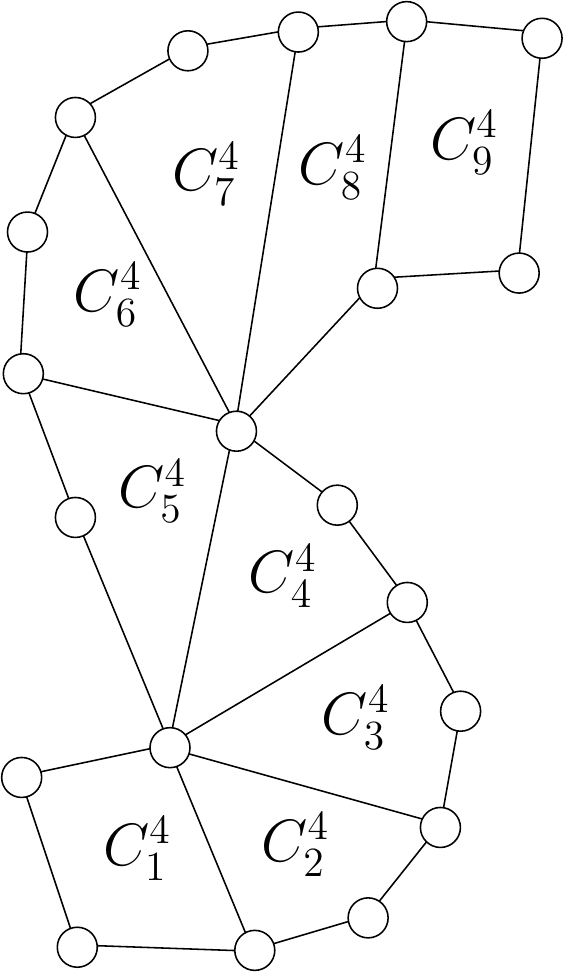}
        \caption{}
        \label{fig:acc}
    \end{subfigure}
    \caption{Examples of $(5,7)$-accordion and $(4,9)$-accordion}
	\label{knladder}
\end{figure}

Let us now define a subclass of $A(k,n)$ which we call \textsl{$(k,n)$-ladders}. It is defined the same way as $(k,n)$-accordions, except that in the construction scheme for a $(k,n)$-ladder, we designate an edge as `free' if both its end points are the new nodes introduced. (Note that for the construction of a $(k,n)$-accordion, an edge is designated as free if at least one of its end points is a new node.)  We denote by $L(k,n)$ the set of all $(k,n)$-ladders.

Note that the $(5,7)$-accordion in Figure~\ref{knladder}\subref{fig:ladd} is also a $(5,7)$-ladder, while the $(4,9)$-accordion in Figure~\ref{knladder}\subref{fig:acc} is not a $(4,9)$-ladder. It is easy to verify that a fan is a $(3,n)$-accordion but not a $(3,n)$-ladder, and $L_n$ is a $(4,n)$-ladder and it is unique.

\medskip
The formulas for counting the number of spanning trees of $F_n$ and
$L_n$ are already known~\cite{Hil74, Sed69}. Now we generalize these results by deriving a formula for counting the number of spanning trees in  $(k,n)$-accordions and $(k,n)$-ladders.

Let $\tau(G)$ be the number of spanning trees
of a graph $G$. Then for any $e \in E(G)$ the following property holds.

\begin{lemma} \label{rr}
$\tau(G)=\tau(T-e)+\tau(G/e)$, where $G-e$ is the graph obtained by
deleting $e$ from $G$, and $G/e$ is obtained by coalescing the
endpoints of $e$.
\end{lemma}

The proof is straightforward due to the fact that the total number of spanning trees in $G$ is the number of spanning trees that contain $e$, plus the number of spanning trees without $e$.

A recursive formula for the number of spanning trees of a $(k,n)$-accordion is given in the following theorem.

\begin{theorem} \label{enu}
	Every $(k,n)$-accordion has the same number of spanning trees. If we denote this number by $\tau(A(k,n))$, then for every integers $k\geq 3$, $n\geq 3$
\begin{equation}\label{eqThm}
	\tau(A(k,n))=k \cdot \tau(A(k,n-1))-\tau(A(k,n-2)).
\end{equation}
\end{theorem}
\begin{proof}
Let $A$ be a $(k,n)$-accordion generated by $k$-cycles $C^k_i,\ldots,C^k_n$. Similarly let $B$ and $C$ be the corresponding $(k,n-1)$ and $(k,n-2)$-accordions generated respectively by the $k$-cycles $C^k_1,\ldots,C^k_{n-1}$ and $C^k_1,\ldots,C^k_{n-2}$. An edge $e$ in $A$ is called a `free edge' if $e\in C^k_n$ and $e\notin C^k_{n-1}$. Likewise, an edge $e$ in $B$ is called a `free edge' if $e\in C^k_{n-1}$ and $e\notin C^k_{n-2}$.  Let $A^1$ be the graph obtained by contracting a free edge of $A$ and $B^1$ be the graph obtained by contracting a free edge of $B$. If $e$ is a free edge of $B$ then $\tau(B-e)=\tau(C)$. Thus from Lemma \ref{rr},
\begin{equation}\label{eq1}
\tau(B) = \tau(C)+\tau(B^1).
\end{equation}
Note that any spanning tree of $A$ either contains all free edges or does not contain exactly one free edge. Further, any spanning tree that contains all free edges does not contain the edge $e^*$ which is common to both $C^k_n$ and $C^k_{n-1}$. Then the graph obtained from $A$ by deleting $e^*$ and then contracting the path of free edges is isomorphic to $B^1$. Since $A$ contains $k-1$ free edges, we have
\begin{equation}\label{eq2}
\tau(A) = (k-1)\tau(B) +\tau(B^1).
\end{equation}
From (\ref{eq1}) and (\ref{eq2}) we have
\begin{equation}\label{eq3}
\tau(A) = k\tau(B) - \tau(C).
\end{equation}
Recall that $A$ is an arbitrary $(k,n)$-accordion with $n\geq 3$, and $B$, $C$ depend on $A$, and are $(k,n-1)$-accordion and $(k,n-2)$-accordion, respectively. Note that there is only one $(k,1)$-accordion and all $(k,2)$-accordions are isomorphic. Hence, for a fixed $k$ and $n\leq2$, every $(k,n)$-accordion has the same number of spanning trees. Then from recursion \eqref{eq3} it follows that for every fixed $k\geq 3$ and $n\geq 1$, the number of spanning trees for any $(k,n)$-accordion is the same, and hence its recursive formula is given by \eqref{eqThm}.
\end{proof}

Theorem~\ref{enu} gives an implicit formula for the number of spanning trees on general $(k,n)$-accordions. In the case $k=3$, it gives us $\tau(A(3,n))=3 \cdot\tau(A(3,n-1))-\tau(A(3,n-2)).$ By solving the recursion we obtain
\begin{equation}\label{eq4}
	\tau(A(3,n))=\frac{1}{\sqrt{5}}\left\{\left(\frac{3+\sqrt{5}}{2}\right)^{n+1}
-\left(\dfrac{3-\sqrt{5}}{2}\right)^{n+1}\right\}.
\end{equation}
When $k=4$, $\tau(A(4,n))=4 \cdot \tau(A(3,n-1))-\tau(A(4,n-2))$, from which it follows that
\begin{equation}\label{eq5}
	\tau(A(4,n))=\frac{\sqrt{3}}{6}\left\{\left(2+\sqrt{3}\right)^{n+1} -\left(2-\sqrt{3}\right)^{n+1}\right\}.
\end{equation}
Formulas \eqref{eq4} and \eqref{eq5} are consistent with the known spanning tree enumeration formulas for fans and ladders~\cite{Hil74, Sed69}, moreover, they generalize them. Furthermore, Theorem~\ref{enu} can be used to deduce explicit formulas for the $(k,n)$-accordions for any fixed $k$. Since every element of $A(k,n)$ contains an exponential number of spanning trees, solving the QMST and its variations on this class of graphs by complete enumeration would be computationally expensive.

\section{Complexity of the QMST variations on some sparse graphs}\label{secComplexity}

We now investigate the complexity of the QMST and its variations on the sparse graphs discussed in Section~\ref{secGraphs}.

\subsection{Intractability results}\label{subsItrac}
\medskip
Recall from Section~\ref{secIntro} that FSTAC is the feasibility version of the adjacent only minimum spanning tree problem with conflict pair constraints.
\begin{theorem}\label{base}
The FSTAC on fan-stars is NP-complete.
\end{theorem}

\begin{proof}
We reduce the 3-SAT problem to the FSTAC on the fan star $FS_n$.

Let
$s= (x_{11}\vee x_{12}\vee x_{13})\wedge(x_{21}\vee x_{22}\vee x_{23}) \wedge \cdots
\wedge(x_{n1}\vee x_{n2}\vee x_{n3}),$ be an instance of 3-SAT given in conjunctive normal form. From this instance,
we construct a graph $G$ with node set $\{u, v_{11},v_{12},v_{13},\dots,v_{n1},v_{n2},v_{n3}\}$
and edge set $E_1\cup E_2$, where
$E_1=\{e_{ij}=(u,v_{ij}):i=1,\ldots,n,\ j=1,2,3\}$,
$E_2=\{(v_{ij},v_{ij+1}):i=1,\ldots, n,\ j=1,2\}$. As shown in Figure~\ref{figFS}, $G$ is a fan-star.
\begin{figure}[h]
   	\centering
	\includegraphics[scale=0.55]{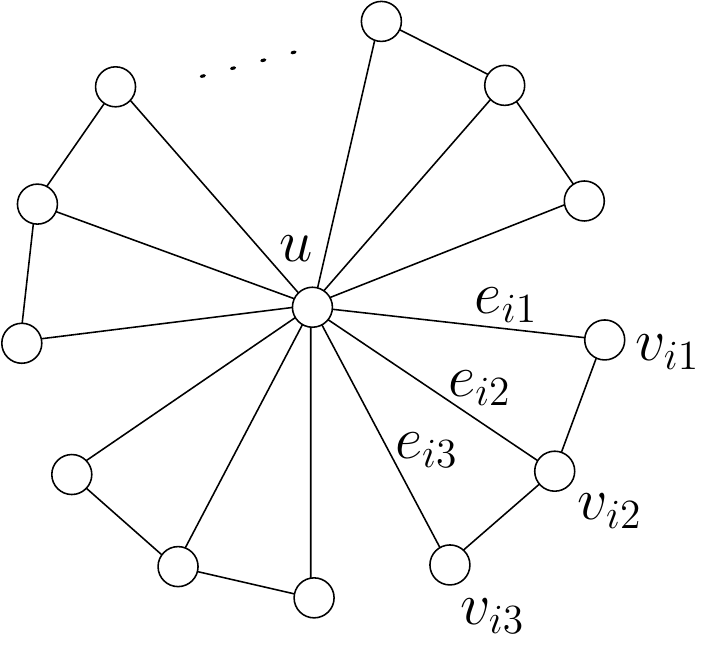}
	\vspace{-5pt}
    	\caption{}
	\label{figFS}
\end{figure}
The conflict set $S$ is defined as $S=\{\{e_{ij}, e_{kl}\}\colon
 x_{ij} \mbox{ and } x_{kl} \mbox{ are negations of each other}\}.$
Note that the edges in each conflict pair are adjacent.

If $T$ is a solution of the FSTAC on $G$, then let
$X_T=\{x_{ij}:i=1,\ldots,n,\ j=1,2,3\}$ be such that $x_{ij}$ is true if
$e_{ij}\in T$, and false otherwise. For each $i=1,\ldots,n$, at least one of
$x_{i1},\ x_{i2},\ x_{i3}$ must be true
since $T$ contains at least one of $e_{i1},\ e_{i2},\ e_{i3}$.
Moreover, $e_{ij}$ and $e_{kl}$ cannot both be in $T$ if
$\{e_{ij},e_{kl}\}\in S$, so in $X_T$ at most one of $x_{ij}$,
$x_{kl}$ will be true if they are negations of each other. Hence
$X_T$ is a true assignment for the 3-SAT problem.

Conversely, suppose $X$ is a true assignment of the 3-SAT problem.
Then $S=\{e_{ij}\colon x_{ij} \text{ is }\allowbreak \text{true}\text{ in } X\}$ is an acyclic
subgraph of $G$ with at most one edge from each conflict pair. If
$S$ spans $G$, then it is a solution of the FSTAC. Otherwise, we add
necessary edges from $E_2$ to $S$ to form a spanning tree $T$, which gives us a
solution of the FSTAC.

The result now follows from the NP-completeness of 3-SAT.
\end{proof}

Since FSTAC is the feasibility version of MSTAC, the MSTAC is also NP-hard on fan-stars. Fan-star is a subgraph of a fan and a wheel, hence the MSTAC on fan-stars can be reduced to the MSTAC on fans or wheels by assigning large costs on additional edges. That proves NP-hardness of the MSTAC on fans, wheels and $(k,n)$-accordions.
MSTAC is a special case of MSTC and AQMST, and MSTC is a special case of QMST, hence all of those problems are NP-hard on fan-stars, fans, wheels and $(k,n)$-accordions. Furthermore, from Theorem~\ref{base} it easily  follows that all bottleneck versions of these problems are NP-hard on fan-stars, fans, wheels and $(k,n)$-accordions. These observations are summarized in the following corollary.

\begin{corollary}\label{corSum}
	MSTAC, MSTC, AQMST, QMST, BSTAC, BSTC, AQBST and QBST are NP-hard on fan-stars, fans, wheels and $(k,n)$-accordions.
\end{corollary}

Next we identify some intractability results for problems on ladders.

\begin{theorem} \label{Nladder}
The FSTC on ladders is NP-complete.
\end{theorem}
\begin{proof}
Again we reduce the 3-SAT problem  to the FSTC on ladders.

Let
$s= (x_{11}\vee x_{12}\vee x_{13})\wedge(x_{21}\vee x_{22}\vee x_{23}) \wedge \cdots
\wedge(x_{n1}\vee x_{n2}\vee x_{n3}),$ be an instance of 3-SAT given in conjunctive normal form. From this,
we construct a ladder $G$ shown in Figure~\ref{figLadNP}. Let $S=\{\{e_{ij}, e_{kl}\}: x_{ij}
\mbox{ and } x_{kl} \mbox{ are negations of each other}\}$. Note that in $S$ the conflict edges are not necessarily adjacent.
\begin{figure}[h]
   \centering
	\includegraphics[scale=0.35]{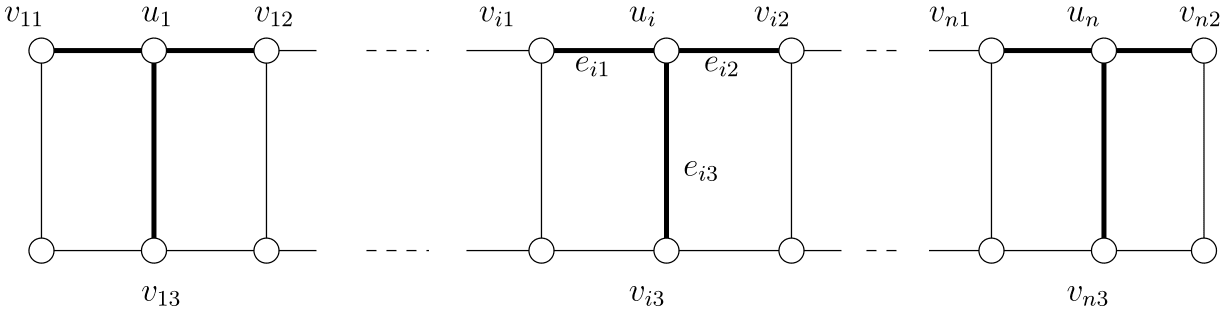}
    \caption{A ladder constructed from a 3-SAT instance}
	\label{figLadNP}
\end{figure}
Then there exists a solution $T$ of the FSTC on $G$, if and only if
the 3-SAT instance has a true assignment. The detailed proof is very
similar to the one given in Theorem~\ref{base}, and hence omitted.
\end{proof}
Again, NP-completeness of FSTC is propagated to MSTC, QMST, BSTC and QBST, as is summarized in the following corollary.

\begin{corollary}\label{corLadSum}
	MSTC, QMST, BSTC and QBST are NP-hard on ladders and $(k,n)$-ladders.
\end{corollary}

\subsection{Polynomially Solvable Special Cases}
\medskip

 Let us now examine the complexity of the QMST variations that are not covered in Section~\ref{subsItrac}. We show that these remaining problems are easy by proposing a linear time algorithm to solve the AQMST on $(k,n)$-ladders. Let us first introduce some notations.

Recall that $(k,n)$-ladder is a graph which is a sequence of $n$ $k$-cycles such that only consecutive $k$-cycles ($C^k_i$ and $C^k_{i-1}$) intersect, and their intersection is an edge (and its corresponding two vertices). Let us label the edges of $C^k_{i}$ by $e^i_1,e^i_2,\ldots,e^i_k$, where $e^i_1$ is the edge that is also in $C^k_{i-1}$, $e^i_k$ is the edge that is also in $C^k_{i+1}$, and $e^i_{k-1}$ and $e^i_{k-2}$ are adjacent to $e^i_k$. Note that then $e^i_k=e^{i+1}_1$. Furthermore, let $v^i_1$ and $v^i_2$ denote the vertices incident to $e^i_k$. See Figure~\ref{ladNotation} for an illustration.
\begin{figure}[h]
	\centering
	\includegraphics[scale=0.63]{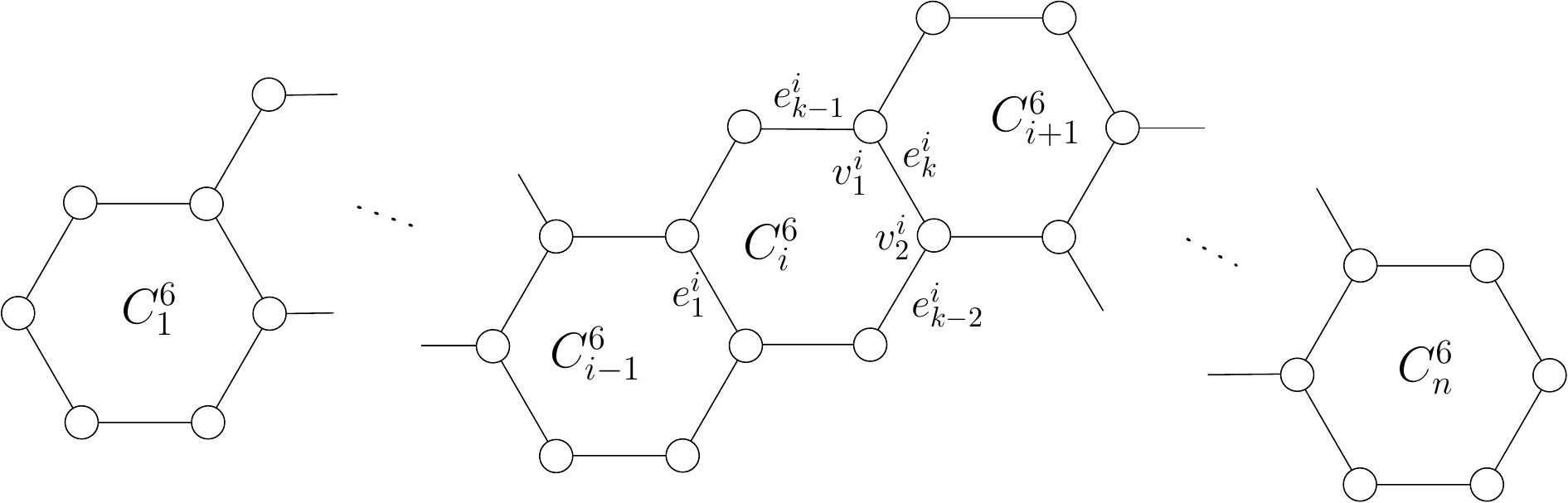}
    	\caption{}
	\label{ladNotation}
\end{figure}

Given a $(k,n)$-ladder $L$, we denote with $L^i$ the $(k,i)$-ladder determined by the first $i$ $k$-cycles of $L$. Next we define $T^i_1$, $T^i_2$, $T^i_3$ and $T^i_4$ to be the spanning trees of $L^i$ with
minimum AQMST cost that satisfy additional properties. In particular, let $T^i_1$ be a
 spanning tree of $L^i$ with minimum AQMST cost that contains $e^i_k$,  $e^i_{k-1}$
and does not contain $e^i_{k-2}$. Analogously, let $T^i_2$ be the minimum spanning tree that
contains $e^i_k$,  $e^i_{k-2}$ and does not contain $e^i_{k-1}$, let $T^i_3$ be the minimum
spanning tree that contains $e^i_{k-1}$,  $e^i_{k-2}$ and does not contain $e^i_{k}$, and let $T^i_4$
be the minimum spanning tree that contains $e^i_k$,  $e^i_{k-1}$ and $e^i_{k-2}$. Note that all
possible configurations of $e^i_k$,  $e^i_{k-1}$, $e^i_{k-2}$ are covered.

Similarly, we define $F^i_1$, $F^i_2$ and $F^i_3$ to be the minimum cost spanning  forests of $L^i$ made of exactly two trees, such that one tree contains $v^i_1$ and the other tree contains $v^i_2$. In particular, let $F^i_1$ be a minimum cost such forest that contains $e^i_{k-1}$ and does not contain $e^i_{k}$, $e^i_{k-2}$, let $F^i_2$ be a minimum cost such forest that contains $e^i_{k-2}$ and does not contain $e^i_{k}$, $e^i_{k-1}$, and let $F^i_3$ be a minimum cost such forest that contains $e^i_{k-1},$ $e^i_{k-2}$ and does not contain $e^i_{k}$. Note that one tree in forests $F^i_1$ and $F^i_2$ are exactly the single vertices $v^i_{1}$ or $v^i_{2}$.

Next we show that for $i\geq2$, if $z(T^{i-1}_j)$, $T^{i-1}_j$, $j=1,\ldots,4$ and $z(F^{i-1}_j)$,
$F^{i-1}_j$, $j=1,2,3$ are known, then $z(T^{i}_j)$, $T^{i}_j$, $j=1,\ldots,4$
and $z(F^{i}_j)$, $F^{i}_j$, $j=1,2,3$ can be calculate in $O(k)$ time.
For two edge disjoint graphs/edge sets $S_1,$ $S_2$ we define $S_1+S_2$ to be the graph spanned by edges of $S_1$ and $S_2$. Then it is easy to verify that the following recursive relations hold:
\begin{align}	
	z(T^i_1)&=\min_{j=1,\ldots,4} \left\{ z\Big(T^{i-1}_j+ \{ C^k_i-e^i_{1}-e^i_{k-2} \} \Big) \right\}, \label{eqAlign1}\\
	z(T^i_2)&=\min_{j=1,\ldots,4} \left\{ z\Big(T^{i-1}_j+ \{ C^k_i-e^i_{1}-e^i_{k-1} \} \Big) \right\}, \\
	z(T^i_3)&=\min_{j=1,\ldots,4} \left\{ z\Big(T^{i-1}_j+ \{ C^k_i-e^i_{1}-e^i_{k} \} \Big) \right\}, \\
	z(T^i_4)&=\min\Bigg\{ \ \min_{j=1,2,3} \left\{ z\Big(F^{i-1}_j+ \{ C^k_i-e^i_{1} \} \Big) \right\}, \nonumber\\ & \hspace{90pt} \min_{\substack{j=1,\ldots,4\\ \ell=2,\ldots,k-3}}\left\{z\Big(T^{i-1}_j+ \{ C^k_i-e^i_{1}-e^i_{\ell} \} \Big)\right\} \ \Bigg\}, \label{eqAlign2}\\
	z(F^i_1)&=\min_{j=1,\ldots,4} \left\{ z\Big(T^{i-1}_j+ \{ C^k_i-e^i_{1}-e^i_{k}-e^i_{k-2} \} \Big) \right\}, \\
	z(F^i_2)&=\min_{j=1,\ldots,4} \left\{ z\Big(T^{i-1}_j+ \{ C^k_i-e^i_{1}-e^i_{k}-e^i_{k-1} \} \Big) \right\}, \\
	z(F^i_3)&=\min\Bigg\{\  \min_{j=1,2,3} \left\{ z\Big(F^{i-1}_j+ \{ C^k_i-e^i_{1}-e^i_{k} \} \Big) \right\},\nonumber\\ & \hspace{90pt} \min_{\substack{j=1,\ldots,4\\ \ell=2,\ldots,k-3}}\left\{z\Big(T^{i-1}_j+ \{ C^k_i-e^i_{1}-e^i_{k}-e^i_{\ell} \} \Big)\right\} \ \Bigg\}. \label{eqAlign3}
\end{align}

Since adjacencies between edges of $C^k_i$ and graphs $T^{i-1}_j$, $F^{i-1}_j$ are known, the minimization functions above can be calculated easily. Function values in \eqref{eqAlign2} and \eqref{eqAlign3} can be calculated in $O(k)$ time and the remaining values in constant time, provided $z(T_j^{i-1})$'s and $z(F^{i-1}_j)$'s are known.

$T^1_j$, $j=1,\ldots,4$ and $F^1_j$, $j=1,2,3$ are easily calculated, and for $i\geq 2$, $T^i_j$'s and $F^i_j$'s and their costs can incrementally be calculated from $T^{i-1}_j$'s and  $F^{i-1}_j$'s along with their costs. When the value of $i$ increases to $n$, the optimal solution of the AQMST on $L$ is obtained. We call this algorithm the \textsl{AQMST-$(k,n)$-ladder algorithm} and is summarized as Algorithm~\ref{alg:ladder}.

\begin{algorithm}[htb]
\caption{AQMST-$(k,n)$-ladder}
\begin{algorithmic}[1]
	\STATE {\bf Input:} A $(k,n)$-ladder graph $L$ with costs $c_e$ for $e\in
E$ and $q(e,f)$ for $e,f\in E, e\neq f$\;
	\STATE Calculate $z(T^1_j)$, $j=1,\ldots,4,$ and $z(F^1_j)$, $j=1,2,3$\;
	\FOR{$i=2$ \TO $n$ }
		\STATE{Determine $T^i_j$, $j=1,\ldots,4,$ and $F^i_j$, $j=1,2,3,$ using  \eqref{eqAlign1}-\eqref{eqAlign3}\;}
	\ENDFOR
	\STATE $T^n$ is the minimum cost tree from $\{T^n_1,\ T^n_2,\ T^n_3,\ T^n_4\}$\;
	\STATE Output $T^n$ and $z(T^n)$;
\end{algorithmic}
\label{alg:ladder}
\end{algorithm}

\begin{theorem}\label{al}
The AQMST-$(k,n)$-ladder algorithm solves AQMST on $(k,n)$-ladders in $O(kn)$ time.
\end{theorem}
\begin{proof}
The correctness of the AQMST-$(k,n)$-ladder algorithm follows from the exploration of all the possible cases resulting in recursion relations \eqref{eqAlign1}-\eqref{eqAlign3}.
There are $O(n)$ iterations of line 4 of the algorithm, and each iteration takes $O(k)$ time to calculate the trees and corresponding costs as discussed above. Hence,  the overall complexity is $O(kn)$.
\end{proof}

The MSTAC can be easily reduced to the AQMST, and by calculating the maximum rather than the summation, Algorithm~\ref{alg:ladder} can be adapted to solve the bottleneck versions of the problems on $(k,n)$-ladders. So without describing the detailed steps, the following corollary holds.

\begin{corollary}
MSTAC, AQBST and BSTAC on $(k,n)$-ladders can be solved in $O(kn)$ time.
\end{corollary}

We have determined complexities of all problem variations on all graph classes investigated. The results are summarized in Table~\ref{table}, in which ``$\times$" represents NP-hardness and ``$\checkmark$" means polynomially solvable.

\begin{table}[ h ]\footnotesize
	\centering
	\begin{tabular}{|c||c|c|c|c|c|c|c|c|}
	\hline
 		&MSTAC&MSTC&AQMST&QMST&BSTAC&BSTC&AQBST&QBST\\\hline
	\hline
	fan-star& $\times$ &$\times$ &$\times$&$\times$&$\times$&$\times$&$\times$&$\times$\\
	fan&  $\times$ &$\times$ &$\times$&$\times$&$\times$&$\times$&$\times$&$\times$\\
	wheel& $\times$ &$\times$ &$\times$&$\times$&$\times$&$\times$&$\times$&$\times$\\
	ladder& $\checkmark$ &$\times$ &$\checkmark$ &$\times$&$\checkmark$ &$\times$ &$\checkmark$ &$\times$\\
	$(k,n)$-ladder& $\checkmark$ &$\times$ &$\checkmark$ &$\times$&$\checkmark$ &$\times$ &$\checkmark$ &$\times$\\
	$(k,n)$-accordion& $\times$ &$\times$ &$\times$&$\times$&$\times$&$\times$&$\times$&$\times$\\\hline
\end{tabular}
	\caption{The polynomial solvability of QMST variations}
	\label{table}
	\vspace{-8pt}
\end{table}

\section{The QMST with row graded cost matrix}\label{secGraded}

In Section~\ref{secComplexity} it is shown that the QMST and its variations are mostly NP-hard even when restricted to very simple classes of graphs on which a wide variety of hard optimization problems can be solved efficiently. Therefore, we shift the focus from special graphs to specially structured cost matrices.
\smallskip

For each $e_i\in E(G)$ consider the  minimum spanning tree problem MST($i,Q$):
\begin{tabbing}
\hspace{3cm}
\=xxxx\=xxxxx\=xxxxx\=xxxxx\=xxxxx\=xxxxx\=xxxx\=xxxx\=xxx\kill
MST$(i,Q)$:\> Minimize $\sum\limits_{e_j\in E(T)}q(i,j)$\\
\> Subject to\\
\>\>\> $T\in \mathcal{F}$,
\end{tabbing}
where $\mathcal{F}$ is the family of all spanning trees of $G$. Let $z^i$ be the optimal objective function value of MST($i,Q$), and consider the minimum spanning tree problem:
\begin{tabbing}
\hspace{5cm}
\=xxxx\=xxxxx\=xxxxx\kill
MST($Q$):\> Minimize $\sum\limits_{e_j\in E(T)}z^j$\\
\> Subject to\\
\>\>\> $T\in \mathcal{F}$.
\end{tabbing}

Let $\tilde{L}$ be the optimal objective function value of the MST($Q$).
It is shown in \cite{Assad} that $\tilde{L}$ is a lower bound for the optimal objective function value of QMST. We call $\tilde{L}$ \textsl{the natural lower bound} for the QMST. If we could find a spanning tree of $G$ with objective function value $\tilde{L}$, it is surely an optimal solution of the QMST.

An $m\times m$ matrix $Q=(q(i,j))$ is said to be \textsl{row graded}, if $q(i,1)\leq q(i,2)\leq \cdots \leq q(i,m)$
for all $i=1,\ldots,m$. Furthermore, $Q$ is called \textsl{doubly graded} if both $Q$ and $Q^T$ are row graded. 
Given an $m\times m$ matrix $Q$ and a permutation $\pi$ on $\{1,2,\ldots,m\}$, we define $\pi(Q)$ to be the $m\times m$ matrix which $(i,j)$-th entry is $q(\pi^{-1}(i),\pi^{-1}(j))$. We say that $Q$ is \textsl{permuted row graded} or \textsl{permuted doubly graded} if there exist a permutation $\pi$ such that $\pi(Q)$ is row graded or doubly graded, respectively. Note that permuted row graded and permuted doubly graded matrices are recognizable in polynomial time.

Let $G$ be a graph with $n$ vertices and edge set $E(G)=\{e_1,e_2,\ldots,e_m\}$.
Given a permutation on $\{1,2,\ldots,m\}$ $\pi$, a spanning tree $T=\{e_{i_1},e_{i_2},\ldots, e_{i_{n-1}}\}$ of $G$ is called \textsl{the $\pi$-critical spanning tree} if the set $\{\pi(i_1),\pi(i_2),\ldots,\pi(i_{n-1})\}$ is lexicographically smallest among all spanning trees of $G$. 
Recall that any MST can be solved by a greedy algorithm, therefore the $\pi$-critical spanning tree is optimal if permuting the cost vector with $\pi$ makes it nondecreasing.

\begin{lemma}\label{nnl}
Let $Q$ be a cost matrix of the QMST on a graph $G$. If $Q$ is permuted row graded such that $\pi(Q)$ is row graded, and if the $\pi$-critical spanning tree $T^0$ is a minimum spanning tree on $G$ with edge costs $z^i:=\sum_{e_j\in T^0}q(i,j)$ for $i=1,\ldots, m$, then $T^0$ is an optimal solution of the QMST on $Q$.
\end{lemma}

\begin{proof}
Since $\pi(Q)$ is row graded, the $\pi$-critical spanning tree $T^0$ is an optimal solution of the MST($i,Q$) for all
$i=1,\ldots,m$ with corresponding optimal objective function value $z^i=\sum_{e_j\in T^0}q(i,j)$. As $T^0$ is also optimal for the MST with edge costs $z^i$, it is optimal for the MST($Q$). Thus $z(T^0)=\tilde{L}$ and the optimality of $T^0$ for the QMST is proved.
\end{proof}

Using Lemma~\ref{nnl} we show that the QMST on permuted doubly graded matrices is polynomially solvable.

\begin{theorem}\label{corM}
If the cost matrix $Q$ of the QMST on a graph $G$ is permuted doubly graded such that $\pi(Q)$ is doubly graded, then the $\pi$-critical spanning tree $T^0$ is an optimal solution.
\end{theorem}
\begin{proof}
Let $z^i:=\sum_{e_j\in T^0}q(i,j)$ for $i=1,\ldots, m$. Since $\pi(Q)$ is row graded, $T^0$ is optimal for MST($i,Q$), $i=1,\ldots,m,$ with corresponding optimal objective function value $z^i$. Also $\pi(Q)^T$ is
row graded, so for all $j$, $q(\pi^{-1}(1),j)\leq q(\pi^{-1}(2),j)\leq \cdots \leq q(\pi^{-1}(m),j)$. 
Then $\sum_{e_j\in T^0}q(\pi^{-1}(1),j)\leq \sum_{e_j\in T^0}q(\pi^{-1}(2),j)\leq \cdots \leq \sum_{e_j\in T^0}q(\pi^{-1}(m),j)$, i.e.\@ $z^{\pi^{-1}(1)}\leq z^{\pi^{-1}(2)}\leq \cdots \leq z^{\pi^{-1}(m)}$. Thus $T^0$ is a minimum spanning tree on $G$ with edge costs $z^i$. From Lemma~\ref{nnl}, $T^0$ is optimal for the
QMST, with optimal objective function value $z(T^0)=\tilde{L}$.
\end{proof}

In the rest of this section, we extend the above results to a more general structure called matroid bases and give a new characterization of matroids in terms of quadratic objective function. To the best of our knowledge, no characterization of matroids is known that uses an optimization problem with a quadratic objective function. 

Let $E=\{1,2,\ldots, m\}$ be a ground set and $\mathcal{F}$ be a family of subsets of $E$ that we call \textsl{bases}, where $|S|=s$, a constant for any $S\in \mathcal{F}$. Let $I=\{X\subseteq S: S\in \mathcal{F}\}$. We call the structure $(E,I)$ an \textsl{independence system} and the structure $(E,\mathcal{F})$ a \textsl{base system}. An independence system $(E,I)$ is called a \textsl{matroid} if and only if for any
$S_1,S_2\in \mathcal{F}$ and $i\in S_1\setminus S_2$, there exists $j\in S_2\setminus S_1$ such that
$S_1\setminus\{i\}\cup\{j\}\in \mathcal{F}$. 
For instance, when $E$ is the edge set of a graph $G$, $\mathcal{F}$ is the collection of
all the spanning trees of $G$, and $I$ is the collection of all acyclic subgraphs of $G$, then $(E,I)$ is called \textsl{the graphic matroid}.

Given a base system $(E,\mathcal{F})$ 
and a weight $w(i,j)$ for each $i,j\in E$, 
the \textsl{quadratic minimum weight base problem (QMWB)} is formulated as follows:
\begin{tabbing}
\hspace{3cm}
\=xxxx\=xxxxx\=xxxxx\=xxxxx\=xxxxx\=xxxx\=xxxx\kill
\> Minimize $\Pi(S)=\sum\limits_{i\in S}\sum\limits_{j\in S}w(i,j)$\\
\> Subject to\\
\>\>\> $S\in \mathcal{F}$,
\end{tabbing}
where $W=(w(i,j))_{m\times m}$ is the associated cost matrix. 
For each $i\in E$ we define a \textsl{minimum weight  base problem} as follows:
\begin{tabbing}
\hspace{3cm}
\=xxxx\=xxxxx\=xxxxx\=xxxxx\=xxxxx\=xxxx\=xxxx\kill
MWB($i,W$):\> Minimize $\sum\limits_{j\in S}w(i,j)$\\
\> Subject to\\
\>\>\> $S\in \mathcal{F}$.
\end{tabbing}
Let $f^i$ be the optimal objective function value of the MWB($i,W$), and similar to the case for QMST, the optimal objective function value $\bar{L}$ of the problem
\begin{tabbing}
\hspace{3cm}
\=xxxx\=xxxxx\=xxxxx\kill
MWB(Q):\> Minimize $\sum\limits_{j\in S}f^j$\\
\> Subject to\\
\>\>\> $S\in \mathcal{F}$,
\end{tabbing}
is called \textsl{the natural lower bound} of the QMWB with cost matrix $W$.

Given a permutation $\pi$ on $\{1,2,\ldots,m\}$, we say that a base $B^0=\{{i_1},{i_2},\ldots, {i_s}\}$ is \textsl{the $\pi$-critical base} if the set $\{\pi(i_1),\pi(i_2),\ldots,\pi(i_s)\}$ is lexicographically smallest among all sets $\{\pi(j_1),\pi(j_2),\ldots,\pi(j_s)\}$ where $\{{j_1},{j_2},\ldots, {j_s}\}\in \mathcal{F}$.

\begin{theorem}\label{matroid}
	The following statements are equivalent:
	\begin{enumerate}[(i)]
		\item $(E,I)$ is a matroid.
		\item Let $W$ be a cost matrix of the QMWB on a base system $(E,\mathcal{F})$. If $W$ is permuted row graded such that $\pi(W)$ is row graded, and if the $\pi$-critical base $S^0$ is a minimum weight base for costs $f^i:=\sum_{j\in S^0}w(i,j)$ for $i=1,\ldots, |E|$, then $S^0$ is an optimal solution of the QMWB with cost matrix $W$.
		\item If the cost matrix $W$ of the QMWB problem on a base system $(E,\mathcal{F})$ is permuted doubly graded such that $\pi(W)$ is doubly graded, then the $\pi$-critical base $S^0$ is an optimal solution. Moreover, the natural lower bound is the optimal objective function value.
	\end{enumerate}
\end{theorem}

\begin{proof}
Using the fact that a minimum weight matroid can be found by the greedy algorithm, statements $(i)$ implies $(ii)$ and $(i)$ implies $(iii)$ can be proved similarly as Lemma~\ref{nnl} and Theorem~\ref{corM}. Hence their proofs are omitted.

To show $(iii)$ implies $(i)$, we assume $(E,I)$ is not a matroid and we aim to show that in that case $(iii)$ is not true. Hence, we assume that there are
$S_1\in \mathcal{F},\ S_2\in \mathcal{F}$ and $K\in S_1\setminus S_2$ such that $S_1\setminus \{K\}\cup \{i\}\notin \mathcal{F}$ for all $i\in S_2\setminus S_1$. Let $\pi$ be a permutation on $\{1,2,\ldots,m\}$ for which $\{\pi(i)\colon i\in S_1\setminus \{K\}\}=\{1,2,\ldots,s-1\}$,  $\{\pi(i)\colon i\in S_2\setminus S_1\}=\{s,s+1,\ldots,|S_1\cup S_2|-1\}$, $\pi(K)=|S_1\cup S_2|$ and $\{\pi(i)\colon i\in E\setminus(S_1\cup S_2)\}=\{|S_1\cup S_2|+1,\ldots,m\}$. Note that $S_1$ is the $\pi$-critical base, since $S_1\setminus \{K\}\cup \{i\}\notin \mathcal{F}$ for all $i\in S_2\setminus S_1$. 
Let a cost matrix $W$ be such that its entries are
\begin{equation*}
	w(i,j)=\begin{cases}
		1 & \mbox{ if } i \mbox{ or } j \mbox{ is in } E\setminus (S_1 \cup S_2),\\
		1 & \mbox{ if } i\in (S_2\setminus S_1)\cup\{K\} \mbox{ and } j=K,\\
		0 & \mbox{ otherwise.}
	\end{cases}
\end{equation*}
Clearly $\pi(W)$ is doubly graded, and hence $W$ is permuted doubly graded, see Figure~\ref{W}.
\begin{figure}[h]
	\vspace{-8pt}
	\begin{align*}
		& \hspace{10pt} \pi(S_1\setminus\{K\}) \hspace{8pt} \pi(S_2\setminus S_1) \hspace{6pt} \pi(K) \hspace{5pt} \pi(E\setminus (S_1 \cup S_2))\\
	\pi(W)=&\left(\begin{array}{ccc|ccc|c|ccc}
		 &  & & &  & &  &  &  & \\
		\ &\ \ 0 \ &\ \ & \ \ \ & 0 &\ \  \ &\hspace{5pt}0\hspace{5pt} & \ \ \ \ \ \ & 1 &\ \ \ \ \ \\
		&  & & &  & &  &  &  & \\
		\hline  &  & & &  & &  &  &  & \\
		\ & 0 & \ & \ & 0 & \ & 1 & \ & 1 & \ \\
		&  & & &  & &  &  &  & \\
		\hline &  & & &  & &  &  &  & \vspace{-5pt}\\
		& 0 &  & & 0 &  & 1 &  & 1 & \\[4pt]
		\hline  &  & & &  & &  &  &  & \\
		 & 1 &  &  & 1 &  & 1 &  & 1 & \ \\
		&  & & &  & &  &  &  & \\
	\end{array}\right)\begin{array}{l}
		 \\
		\pi(S_1\setminus\{K\}) \\
		\\
		\\ 
		\pi(S_2\setminus S_1) \\
		\vspace{5pt}\\
		\pi(K) \\
		\\
		\pi(E\setminus (S_1 \cup S_2)) \\
		\\
	\end{array}
	\end{align*}
	\vspace{-8pt}\caption{}
	\label{W}
\end{figure}
Now let us consider the objective function values of $S_1$ and $S_2$ in QMWB with cost matrix $W$. $\Pi(S_1)=1$ and $\Pi(S_2)=0$.
Hence, the $\pi$-critical base $S_1$ is not an optimal solution of the QMWB on cost matrix $W$, which implies that $(iii)$ is not true.

Similarly, we prove that $(ii)$ implies $(i)$ by showing that if $(E,I)$ is not a matroid, then $(ii)$ is not true. That will complete the proof of the theorem. So, assume that there are
$S_1\in \mathcal{F},\ S_2\in \mathcal{F}$ and $K\in S_1\setminus S_2$ such that $S_1\setminus \{K\}\cup \{i\}\notin \mathcal{F}$ for all $i\in S_2\setminus S_1$. Then again consider permutation $\pi$ and the cost matrix $\pi(W)$ from Figure~\ref{W}. $\pi(W)$ is row graded and the $\pi$-critical base $S_1$ is a minimum weight base on the costs $f^i=\sum_{j\in S_1}w(i,j)$, but $S_1$ is not an optimal solution of the corresponding QMWB, $S_2$ is.
\end{proof}

We end this section by noting that Theorem~\ref{corM} and Theorem~\ref{matroid} hold true also for the QBST and \textsl{quadratic bottleneck base problem}, respectively. That is, when the sum in the objective value function is replaced by the maximum. The same proofs work, since the greedy algorithm obtains an optimal solution also for the linear bottleneck objective functions on a base system of a matroid \cite{GP90}.


\section*{Acknowledgments}

This work was supported by an NSERC discovery grant and an NSERC discovery accelerator supplement awarded to Abraham P. Punnen.

\end{document}